\documentclass[letterpaper, 10 pt, conference]{ieeeconf}
\usepackage[T1]{fontenc}
\IEEEoverridecommandlockouts                          

\usepackage{amsmath} 
\usepackage{amssymb} 
\usepackage{xcolor}
\usepackage{caption}
\usepackage{graphicx}
\usepackage{mathtools}
\usepackage{caption}
\usepackage{floatrow}
\usepackage{tikz}
\usepackage{tikz-cd}
\usepackage{url}
\usepackage{tabularx}

\usepackage{algorithmicx}
    \usepackage[ruled]{algorithm}
    \usepackage{algpseudocode}

\usepackage{pgfplots}
\pgfplotsset{compat=1.18}
\usepackage{xcolor}
\usetikzlibrary{matrix,arrows,calc,positioning,shapes,decorations.pathreplacing}
\usepackage{graphicx}

\usepackage{amsmath} 

\usepackage{enumerate}
\usepackage[all,tips]{xy}
\SelectTips{cm}{11}

\newtheorem{theorem}{Theorem}

\newtheorem{definition}{Definition}
\newtheorem{remark}{Remark}
\newtheorem{proposition}{Proposition}

\newtheorem{assumption}{Assumption}

\newcommand{\R}{\mathbb{R}}

\usepackage{amssymb}

\newcommand{\proa}{A^*G \mbox{$\;$}_{\tau^*} \kern-3pt\times_\alpha
G \mbox{$\;$}_\beta \kern-3pt\times_{\tau^*} A^*G}

\newcommand{\D}{\mathcal{D}}

\newcommand\tran{\mkern-2mu\raise1.25ex\hbox{$\scriptscriptstyle\top\hspace{0.5mm}$}\mkern-3.5mu}

\usepackage[noabbrev]{cleveref} 
\crefname{rem}{Remark}{Remarks}
\crefname{exam}{Example}{Examples}
\crefname{assum}{Assumption}{Assumptions}
\crefname{prop}{Proposition}{Propositions}
\crefname{propy}{Property}{Properties}
\crefname{cor}{Corollary}{Corollaries}
\crefname{lem}{Lemma}{Lemmas}
\crefname{section}{Section}{Sections}
\crefname{thm}{Theorem}{Theorems}
\crefname{alg}{Algorithm}{Algorithms}
\crefname{defn}{Definition}{Definitions}
\crefname{figure}{Fig.}{Fig.}
\Crefname{figure}{Figure}{Figures}
\crefname{equation}{}{}

\IEEEoverridecommandlockouts                              

\title{ \bf
Structure-Preserving Learning of Nonholonomic Dynamics
}

\author{Thomas Beckers, Anthony Bloch,  Leonardo Colombo
\thanks{Thomas Beckers is with the Department of Computer Science, Vanderbilt University, Nashville, TN 37212, USA {\tt\small thomas.beckers@vanderbilt.edu}}
\thanks{A. Bloch is with Department of Mathematics, University of Michigan, Ann Arbor, MI 48109, USA.{\tt\small abloch@umich.edu}}
\thanks{L. Colombo is with Centre for Automation and Robotics (CSIC-UPM), Ctra. M300 Campo Real, Km 0,200, Arganda
del Rey - 28500 Madrid, Spain.{\tt\small leonardo.colombo@csic.es}} 
\thanks{ L. Colombo acknowledge financial support from Grant PID2022-137909NB-C21 funded by MCIN/AEI/ 10.13039/501100011033. The research leading to these results was supported in part by iRoboCity2030-CM, Robótica Inteligente para Ciudades Sostenibles (TEC-2024/TEC-62), funded by the Programas de Actividades I+D en Tecnologías en la Comunidad de Madrid. A.Bloch was partially supported by NSF grant  DMS-2103026, and AFOSR grants FA
9550-22-1-0215 and FA 9550-23-1-0400.}
}

\begin{document}

\maketitle
\thispagestyle{empty}
\pagestyle{empty}

\begin{abstract}
Data-driven modeling is playing an increasing role in robotics and control, yet standard learning methods typically ignore the geometric structure of nonholonomic systems. As a consequence, the learned dynamics may violate the nonholonomic constraints and produce physically inconsistent motions. In this paper, we introduce a structure-preserving Gaussian process (GP) framework for learning nonholonomic dynamics. Our main ingredient is a nonholonomic matrix-valued kernel that incorporates the constraint distribution directly into the GP prior. This construction ensures that the learned vector field satisfies the nonholonomic constraints for all inputs. We show that the proposed kernel is positive semidefinite, characterize its associated reproducing kernel Hilbert space as a space of admissible vector fields, and prove that the resulting estimator admits a coordinate representation adapted to the constraint distribution. We also establish the consistency of the learned model. Numerical simulations on a vertical rolling disk illustrate the effectiveness of the proposed approach.
\end{abstract}

\section{Introduction}

Data-driven modeling has become an increasingly important tool in robotics, control, and dynamical systems. In many applications, the dynamics of a system are not perfectly known and must be learned from experimental data. Gaussian process (GP) regression has emerged as a powerful nonparametric approach for learning unknown dynamics due to its flexibility, probabilistic interpretation, and strong theoretical guarantees \cite{rasmussen2006gaussian}. However, standard machine learning methods typically ignore the geometric structure underlying many mechanical systems. 

Nonholonomic systems arise frequently in robotics and vehicle dynamics, including wheeled mobile robots and robotic locomotion systems \cite{bloch2003nonholonomic}, \cite{monforte2002geometric}, \cite{ne_mark2004dynamics}. These systems are characterized by velocity constraints that restrict the admissible directions of motion. As a consequence, their dynamics evolve on a distribution of allowable velocities rather than on the full tangent bundle of the configuration space. When learning dynamics directly from data, ignoring these constraints can lead to models that violate the admissible distribution and produce physically inconsistent predictions. This issue is well known in the modeling of robotic systems, where learned models may generate trajectories that are incompatible with the underlying mechanical constraints.

Several works have demonstrated that embedding geometric structure into learning architectures can markedly improve the consistency and interpretability of learned dynamical models. Representative examples include Hamiltonian and Lagrangian neural networks, equivariant models on Lie groups, and symmetry-preserving learning methods for Hamiltonian systems with conserved quantities \cite{Greydanus2019,Jin2020,Zhong2020,Cranmer2020,Finzi2020,Vaquero2024,Eldred2024,beckers2022gaussian}.

More recently, nonholonomic learning has started to receive attention in approaches based on constraint discovery \cite{Wang2024} while nonnholonomic neural networks have been exoplored in \cite{diaz2025lagrangian}. In this work, we incorporate the constraint distribution directly into the Gaussian process prior. In this way, the learned vector field is admissible by construction for every input. In particular, we propose a GP framework for learning nonholonomic dynamics while preserving the constraint distribution. The key idea is to construct a matrix-valued kernel that incorporates the constraint distribution directly into the GP prior. The resulting \emph{nonholonomic kernel} ensures that the learned vector field automatically satisfies the constraints for all inputs.

\textit{Main contributions}: \begin{itemize}
\item We introduce a nonholonomic kernel that embeds the constraint distribution into the GP prior and we prove that the resulting kernel is positive semidefinite and therefore defines a valid Gaussian process model.

\item We characterize the reproducing kernel Hilbert space induced by the kernel and show that it consists only of admissible vector fields.

\item We establish that learning with the nonholonomic kernel is equivalent to Gaussian process regression in coordinates adapted to the constraint distribution.

\item We prove consistency of the resulting estimator.

\end{itemize}

These results provide a structure-preserving learning framework for nonholonomic systems that combines geometric modeling with modern machine learning techniques. 

The paper is structured as follows. In Section~\ref{pre}, we review nonholonomic systems and recall the Gaussian process regression framework used throughout the paper. In Section~\ref{sec:nhkernel}, we introduce the nonholonomic kernel and establish its main structural properties and the analysis of the associated estimator, including its coordinate representation adapted to the constraint distribution and consistency properties. In Section~\ref{sec:simulations}, we present numerical simulations illustrating the proposed method for the vertical rolling disk. Section~\ref{sec:conclusions} contains concluding remarks and directions for future work.

\section{Background}\label{pre}

\subsection{Nonholonomic Mechanics}

Let $Q$ denote the configuration manifold of a mechanical system, with local coordinates $q \in Q$ with dim$(Q)=n$. We consider nonholonomic systems subject to linear velocity constraints of the form
$A(q)\dot q = 0$, where $A(q) \in \mathbb{R}^{k \times n}$ has full row rank and defines $k$ independent constraints. 

These constraints define the subspace
$\mathcal{D}_q = \ker A(q) \subset T_q Q$,
which specifies the set of allowable velocities at configuration $q$. A trajectory $q(t)$ is admissible if and only if $\dot q(t) \in \mathcal{D}_{q(t)}$.
Thus, the constraints determine a rank-$(n-k)$ distribution $\mathcal{D}\subset TQ$, whose fiber at each $q\in Q$ is $\D_q$.

For mechanical systems with Lagrangian $L(q,\dot q)$, the equations of motion subject to the nonholonomic constraints can be written using Lagrange multipliers as $\lambda\in\R^k$, see, for instance,  \cite{bloch2003nonholonomic,monforte2002geometric}
\begin{equation*}
\frac{d}{dt}\frac{\partial L}{\partial \dot q} - \frac{\partial L}{\partial q} = A(q)^\top \lambda, \qquad
A(q)\dot q = 0.
\end{equation*}

In this work we focus on learning admissible dynamics of the form $\dot q = f(q)$, where the vector field satisfies $f(q) \in \mathcal{D}_q$.

To describe the admissible subspace it is convenient to introduce the orthogonal projector onto the constraint distribution. 
For each $q\in Q$, let
\begin{equation}
P(q)=I-A(q)^\dagger A(q),
\label{eq:projector}
\end{equation}
where $A(q)^\dagger$ denotes the Moore--Penrose pseudoinverse. In local coordinates, and with respect to the Euclidean inner product on the ambient coordinate space, $P(q)$ is the orthogonal projector onto the admissible distribution
$\D_q=\ker A(q)$. In particular, $P(q)^2=P(q)$ and $\operatorname{Im}(P(q))=\D_q$.

This representation will be used in the next section to construct Gaussian process kernels that generate vector fields lying in the distribution $\mathcal{D}$.

\subsection{Gaussian Processes and Kernels}\label{sec:gp}
The Gaussian process (GP) provides a nonparametric framework for learning unknown functions from data. A GP is fully specified by a mean function $m(q)$ and a covariance function $k(q,q')$, and is written as $f \sim \mathcal{GP}(m(q),k(q,q'))$. In practice, the mean function is often taken to be zero, which simplifies the resulting expressions without reducing the expressive power of the model. Thus, in this work we consider the prior $f \sim \mathcal{GP}(0,k(q,q'))$, see \cite{rasmussen2006gaussian}.

Given training inputs $q_1,\dots,q_N \in Q$ and corresponding scalar observations $y_1,\dots,y_N \in \mathbb{R}$, the kernel matrix (or Gram matrix) is defined by
$K_{ij}=k(q_i,q_j)$, $i,j=1,\dots,N$. Assuming additive Gaussian measurement noise with variance $\sigma^2$, the posterior mean prediction at a test point $q\in Q$ is given by 
\begin{equation}
\hat f(q)=k(q,Q)\bigl(K(Q,Q)+\sigma^2 I\bigr)^{-1}Y,
\end{equation}
where $Y=[y_1,\dots,y_N]^\top$, $K(Q,Q)=[k(q_i,q_j)]_{i,j=1}^N$ is the Gram matrix associated with the training inputs, and $k(q,Q)=[k(q,q_1),\dots,k(q,q_N)]$.

A symmetric function $k\colon Q\times Q\to\mathbb{R}$ is called a positive definite kernel if, for any points $q_1,\dots,q_N\in Q$ and any coefficients $c_1,\dots,c_N\in\mathbb{R}$, $\displaystyle{\sum_{i,j=1}^N c_i\, k(q_i,q_j)\, c_j \ge 0}$. This condition guarantees that the associated Gram matrix is positive semidefinite and therefore defines a valid covariance function for GP regression. For kernels defined on non-Euclidean domains and Riemannian manifolds see \cite{Jayasumana2015,Borovitskiy2020}.

To learn vector fields on $Q$, it is natural to consider matrix-valued kernels $K\colon Q\times Q\to\mathbb{R}^{n\times n}$. Such a kernel is positive semidefinite if, for any $q_1,\dots,q_N\in Q$ and $c_1,\dots,c_N\in\mathbb{R}^n$,
$$\sum_{i,j=1}^N c_i^\top K(q_i,q_j)c_j \ge 0.$$ Matrix-valued kernels therefore provide a natural framework for GP models of vector-valued dynamics \cite{micchelli2005learning}, \cite{carmeli2010vector}.

Let $Q$ be a compact set. A continuous positive definite kernel $k:Q\times Q\to\mathbb{R}$ is said to be \textit{universal} if its associated RKHS $\mathcal H_k$ is dense in $C(Q)$, the space of real-valued continuous functions on $Q$, with respect to the uniform norm.  Typical examples on compact Euclidean domains include the squared exponential kernel and suitable Mat\'ern kernels; see, e.g., \cite{steinwart2008support}, \cite{sriperumbudur2011universality}.

\section{Learning Nonholonomic Dynamics}\label{sec:nhkernel}

\subsection{Nonholonomic Kernel} 

We now introduce a kernel that incorporates the geometric structure of nonholonomic systems directly into the GP prior.

\begin{definition}
Let $k(q,q')$ be a scalar positive definite kernel on $Q$. The \emph{nonholonomic kernel} associated with the constraint distribution $\D$ is defined as
\begin{equation}
K_{NH}(q,q') = P(q)\, k(q,q')\, P(q').\label{NHK}
\end{equation}
\end{definition}

This kernel produces vector-valued predictions in $TQ$ while enforcing admissibility with respect to the nonholonomic constraint distribution $\mathcal D$.

\begin{definition}A reproducing kernel Hilbert space (RKHS) $\mathcal{H}_K$ associated with a matrix-valued positive semidefinite kernel $K: Q\times Q \to \mathbb{R}^{n\times n}$ is a Hilbert space of vector-valued functions $f:Q\to\mathbb{R}^n$ such that, for every $q\in Q$ and $v\in\mathbb{R}^n$, the function $K(\cdot,q)v$ belongs to $\mathcal{H}_K$ and the reproducing property $\langle f, K(\cdot,q)v\rangle_{\mathcal{H}_K}= \langle f(q), v\rangle_{\mathbb{R}^n}$ holds for all $f\in\mathcal{H}_K$.\end{definition}

Associated with the matrix-valued kernel $K_{\mathrm{NH}}$ is a RKHS $\mathcal H_{\mathrm{NH}}$ of vector-valued functions on $Q$. The following result characterizes $\mathcal H_{\mathrm{NH}}$ and shows that it consists only of vector fields $f$ satisfying $f(q)\in \mathcal D_q$ for all $q\in Q$.

\begin{proposition}
Let $\mathcal H_{kI_n}$ be the RKHS associated with the matrix-valued kernel
$K_0(q,q')=k(q,q')I_n$, and define the operator $(Tg)(q):=P(q)g(q)$.
Then the RKHS $\mathcal H_{\mathrm{NH}}$ associated with the nonholonomic kernel \eqref{NHK} is given by $\mathcal H_{\mathrm{NH}}
=
\{\, Tg : g\in \mathcal H_{kI_n}\,\}$, endowed with the norm $\|f\|_{\mathcal H_{\mathrm{NH}}}
=
\inf\{\|g\|_{\mathcal H_{kI_n}}: Tg=f\}$.

In particular, every $f\in\mathcal H_{\mathrm{NH}}$ can be written as $f(q)=P(q)g(q)$ for some $g\in\mathcal H_{kI_n}$, and hence $f(q)\in \mathcal D_q
\,
\forall q\in Q$.
\end{proposition}

\begin{proof}
Define the linear operator $T:\mathcal H_{kI_n}\to \{f:Q\to\mathbb R^n\}$,
$(Tg)(q):=P(q)g(q)$. We claim that the matrix-valued kernel $K_{\mathrm{NH}}(q,q')$ is precisely the kernel induced by the operator $T$.

Indeed, for any $q'\in Q$ and $c\in\mathbb R^n$, the kernel section of $K_0$ is
$K_0(\cdot,q')c = k(\cdot,q')c \in \mathcal H_{kI_n}$, and applying $T$ gives
\[
T\bigl(K_0(\cdot,q')P(q')c\bigr)(q)
=
P(q)\,k(q,q')\,P(q')c
=
K_{\mathrm{NH}}(q,q')c.
\]
Hence every kernel section of $K_{\mathrm{NH}}$ belongs to $T(\mathcal H_{kI_n})$.

Conversely, the linear span of such sections is contained in $T(\mathcal H_{kI_n})$, since for any finite linear combination,
\[
\sum_{i=1}^N K_{\mathrm{NH}}(\cdot,q_i)c_i
=
T\left(
\sum_{i=1}^N K_0(\cdot,q_i)P(q_i)c_i
\right).
\]
Therefore, the RKHS generated by $K_{\mathrm{NH}}$ is exactly the image of $\mathcal H_{kI_n}$ under $T$, endowed with the standard norm
\[
\|f\|_{\mathcal H_{\mathrm{NH}}}
=
\inf\{\|g\|_{\mathcal H_{kI_n}}: Tg=f\}.
\]

It follows that every $f\in\mathcal H_{\mathrm{NH}}$ admits a representation
$f=Tg$ with $g\in\mathcal H_{kI_n}$, that is, $f(q)=P(q)g(q)$. Since $\operatorname{Im}(P(q))=\mathcal D_q$, then $f(q)\in \mathcal D_q$, $\forall q\in Q$.
\end{proof}
\begin{remark}
The previous proposition shows that the nonholonomic kernel does not merely project the posterior prediction a posteriori. Rather, it restricts the entire GP prior to vector fields taking values in $\D$.
\end{remark}

\begin{proposition}
If $k(q,q')$ is positive definite, then the nonholonomic kernel $K_{NH}(q,q')$ is a positive semidefinite matrix-valued kernel.
\end{proposition}

\begin{proof}
Let $c_1,\ldots,c_N\in\mathbb R^n$ be arbitrary vectors. Then
\begin{equation*}
\sum_{i,j=1}^N c_i^T K_{NH}(q_i,q_j)c_j
=
\sum_{i,j=1}^N c_i^T P(q_i)\, k(q_i,q_j)\, P(q_j)c_j .
\end{equation*}
Defining $\tilde c_i = P(q_i)c_i$ gives
$\displaystyle{\sum_{i,j=1}^N \tilde c_i^T k(q_i,q_j)\tilde c_j \ge 0}$, since $k$ is positive definite.
\end{proof}

Propositions 1 and 2 together show that the proposed nonholonomic kernel defines a valid GP model whose associated hypothesis space consists only of admissible vector fields. We now make this implication explicit at the level of the GP predictor, showing that the posterior mean automatically satisfies the nonholonomic constraints for every input.

\begin{proposition}
Let $\hat f(q)$ be the Gaussian process predictor obtained using the nonholonomic kernel $K_{NH}$. Then $A(q)\hat f(q)=0$ for all $q\in Q$.
\end{proposition}

\begin{proof}
From the definition of the nonholonomic kernel, any mean function generated by the Gaussian process has the form $\displaystyle{\hat f(q)=\sum_{i=1}^N K_{NH}(q,q_i)c_i}$. Substituting the kernel definition gives
$\displaystyle{\hat f(q)=P(q)\sum_{i=1}^N k(q,q_i)P(q_i)c_i}$.

Since $P(q)$ is the orthogonal projector onto $\mathcal D_q$, we have $A(q)P(q)=0$. Therefore $A(q)\hat f(q)=0$. \end{proof}

\subsection{Adapted Coordinates Interpretation} 

Nonholonomic systems admit a natural representation in terms of coordinates adapted to the constraint distribution. Let $B(q)\in\mathbb R^{n\times (n-k)}$ be a smooth matrix whose columns form a basis of the $\mathcal D_q$, i.e., $\mathcal D_q = \mathrm{Im}(B(q))$.

Any admissible velocity can then be written as $\dot q = B(q)\nu$, where $\nu\in\mathbb R^{n-k}$ represents the generalized velocity in adapted coordinates to $\D$. Consequently, any admissible vector field $f(q)$ satisfying $f(q)\in\mathcal D_q$ admits the representation
\begin{equation}
f(q)=B(q)\nu(q),
\end{equation}
for some vector field $\nu(q)$ adapted to $\D$.

Next, we show that the nonholonomic kernel implicitly performs GP regression on these adapted coordinates.

\begin{theorem}\label{th1}
Let $B(q)\in\mathbb{R}^{n\times(n-k)}$ be a smooth matrix whose columns form a basis of the constraint distribution $\D_q$, so that $\mathcal{D}_q=\operatorname{Im}(B(q))$. Then every function $f\in H_{\mathrm{NH}}$ admits a representation
$f(q)=B(q)\nu(q)$, for some vector field $\nu(q)$ on $\D$. Accordingly, learning with the nonholonomic kernel admits an equivalent representation in adapted coordinates to the constraint distribution.
\end{theorem}
\begin{proof}
Any function $f\in H_{\mathrm{NH}}$ can be written as $f(q)=P(q)g(q)$, for some vector-valued function $g$. Since $P(q)$ projects onto $\D_q$ and $\D_q=\operatorname{Im}(B(q))$, there exists a pseudoinverse $B(q)^\dagger$ such that $P(q)=B(q)B(q)^\dagger$. Therefore $f(q)=B(q)B(q)^\dagger g(q)$.

Defining $\nu(q)=B(q)^\dagger g(q)$, we obtain $f(q)=B(q)\nu(q)$. Thus every admissible function generated by the nonholonomic kernel can be represented in adapted coordinates.
\end{proof}

\begin{remark}
The previous theorem should be understood as a representation result at the level of admissible function classes. An exact statistical equivalence between regression in ambient coordinates and regression in adapted coordinates additionally depends on the observation model and on how measurement noise transforms under the map $B(q)^\dagger$.\end{remark}

\begin{remark}Theorem \ref{th1} provides an alternative interpretation of the proposed method. Instead of explicitly parameterizing the admissible dynamics using adapted coordinates $B(q)$, the nonholonomic kernel enforces the same structure implicitly through the projection operator $P(q)$. As a consequence, the Gaussian process prior is automatically restricted to vector fields compatible with the constraint distribution.\end{remark}

\subsection{Consistency of the Estimator}

We now study the statistical properties of the estimator induced by the nonholonomic kernel. In particular, we show that the resulting GP scheme is consistent for the true admissible vector field, provided that the latter belongs to, or is well approximated by, the function class induced by the underlying scalar kernel.

Assume that we observe noisy samples of the vector field $y_i = f^\star(q_i) + \varepsilon_i$,
where $q_i \in Q$, the noise terms $\varepsilon_i$ are independent Gaussian variables with variance $\sigma^2$, and the true dynamics $f^\star$ satisfy the nonholonomic constraints $f^\star(q) \in \mathcal D_q$. Let $\hat f_N$ denote the GP estimator obtained from $N$ samples using the nonholonomic kernel $K_{NH}$.

\begin{proposition}
Let $\hat g_N$ denote an estimator obtained with the kernel $k(q,q')I_n$, and let $\hat f_N$ be the estimator induced by the nonholonomic kernel, so that $\hat f_N(q)=P(q)\hat g_N(q)$.

Assume that the true admissible dynamics satisfy $f^\star(q)=P(q)g^\star(q)$ for some vector-valued function $g^\star$. Then
\begin{equation}
\|\hat f_N(q)-f^\star(q)\|
\le
\|P(q)\|\,\|\hat g_N(q)-g^\star(q)\|.
\label{eq:projection_bound}
\end{equation}
\end{proposition}

\begin{proof}
From the construction of the estimator we have $\hat f_N(q)=P(q)\hat g_N(q)$, and $f^\star(q)=P(q)g^\star(q)$. Therefore
$\hat f_N(q)-f^\star(q)
=
P(q)(\hat g_N(q)-g^\star(q))$. Taking norms, $\|\hat f_N(q)-f^\star(q)\|
\le
\|P(q)\|\,\|\hat g_N(q)-g^\star(q)\|$. \end{proof}

\begin{remark}Proposition 4 shows that the projection step is stable: the error of the nonholonomic estimator is controlled by the error of the underlying unconstrained estimator through the operator norm of $P(q)$. Since $P(q)$ is an orthogonal projector, one has $\|P(q)\|\le 1$ in the induced Euclidean operator norm. Hence $\|\hat f_N(q)-f^\star(q)\|
\le
\|\hat g_N(q)-g^\star(q)\|$.\end{remark}

Next, we address the universal propery of the proposed nonholonomic kernel. Let $Q$ be a compact set. A continuous positive definite kernel $k:Q\times Q\to\mathbb{R}$ is said to be \textit{universal} if its associated RKHS $\mathcal H_k$ is dense in $C(Q)$, the space of real-valued continuous functions on $Q$, with respect to the uniform norm. Typical examples on compact Euclidean domains include the squared exponential kernel and suitable Mat\'ern kernels; see, e.g., \cite{steinwart2008support}, \cite{sriperumbudur2011universality}.

\begin{assumption}
The scalar kernel $k(q,q')$ is universal on the compact set $Q$.
\end{assumption}

\begin{assumption}
The true vector field $f^\star$ belongs to the RKHS $\mathcal H_{NH}$ associated with the nonholonomic kernel.
\end{assumption}

Assumption 2 is a standard realizability assumption in RKHS-based consistency analysis. It requires that the true admissible dynamics lie in the space induced by the nonholonomic kernel, i.e., they can be represented as the projection of a sufficiently regular ambient vector field. While this is not automatic, it is natural for smooth nonholonomic dynamics on compact domains and for sufficiently rich scalar kernels.

\begin{theorem}
Let $\hat f_N$ be the Gaussian process estimator obtained using the nonholonomic kernel $K_{\mathrm{NH}}$ and $N$ distinct training samples. Under the above assumptions, $\hat f_N$ is uniformly consistent in probability, namely, for every $\varepsilon>0$,
\[
\mathbb{P}\!\left(
\sup_{q\in Q}\|\hat f_N(q)-f^\star(q)\|>\varepsilon
\right)\to 0\,\, \text{as } N\to\infty.
\]
\end{theorem}

\begin{proof}
By Proposition 1, the RKHS associated with the nonholonomic kernel is
\[
\mathcal H_{\mathrm{NH}}
=
\left\{
f:Q\to\mathbb R^n
\;\middle|\;
f(q)=P(q)g(q),\ \ g\in\mathcal H_{kI_n}
\right\}.
\]
Since $f^\star\in\mathcal H_{\mathrm{NH}}$, there exists $g^\star\in\mathcal H_{kI_n}$ such that $f^\star(q)=P(q)g^\star(q)$,\,$q\in Q$.

Let $\hat g_N$ denote the Gaussian process estimator associated with the kernel
$K_0(q,q')=k(q,q')I_n$. Under Assumptions 1 and 2 and standard consistency results for kernel-based GP regression, $\hat g_N$ is uniformly consistent in probability for $g^\star$, namely, for every $\varepsilon>0$,
\[
\mathbb{P}\!\left(
\sup_{q\in Q}\|\hat g_N(q)-g^\star(q)\|>\varepsilon
\right)\to 0
\quad\text{as }N\to\infty.
\]

On the other hand, the estimator induced by the nonholonomic kernel satisfies
$\hat f_N(q)=P(q)\hat g_N(q)$. Therefore, for every $q\in Q$, $\hat f_N(q)-f^\star(q)
=
P(q)\bigl(\hat g_N(q)-g^\star(q)\bigr)$,
and hence $\|\hat f_N(q)-f^\star(q)\|
\le
\|P(q)\|\,\|\hat g_N(q)-g^\star(q)\|$.

By construction, $P(q)=I-A(q)^\dagger A(q)$ is the orthogonal projector onto  $\mathcal D_q=\ker A(q)$. In particular, $P(q)$ is self-adjoint and idempotent, that is, $P(q)^\top=P(q)$ and $P(q)^2=P(q)$. Hence, with respect to the Euclidean norm, its operator norm satisfies $\|P(q)\|\le 1$, $q\in Q$. 

Indeed, if $v\in\mathbb R^n$ is decomposed orthogonally as $v=v_{\mathcal D}+v_{\mathcal D^\perp}$,
$v_{\mathcal D}\in\mathcal D_q,\, v_{\mathcal D^\perp}\in\mathcal D_q^\perp$,
then $P(q)v=v_{\mathcal D}$, and therefore
$\|P(q)v\|=\|v_{\mathcal D}\|\le \|v\|.$ Consequently,
\[
\sup_{q\in Q}\|\hat f_N(q)-f^\star(q)\|
\le
\sup_{q\in Q}\|\hat g_N(q)-g^\star(q)\|.
\]

Now let $\varepsilon>0$. Then
\[
\mathbb{P}\!\left(
\sup_{q\in Q}\|\hat f_N(q)\!-\!f^\star(q)\|\!>\!\varepsilon\!
\right)
\!\le\!
\mathbb{P}\!\left(
\sup_{q\in Q}\|\hat g_N(q)\!-\!g^\star(q)\|\!>\!\varepsilon\!
\right)\!.
\]
Since the right-hand side converges to zero as $N\to\infty$, we conclude that, for every $\varepsilon>0$,
\[
\mathbb{P}\!\left(
\sup_{q\in Q}\|\hat f_N(q)-f^\star(q)\|\>\varepsilon
\right)\to 0
\quad\text{as }N\to\infty.
\] \end{proof}

This estimate shows that consistency of the unconstrained estimator $\hat g_N$ transfers directly to the projected estimator $\hat f_N$, provided the target dynamics admit the representation $f^\star(q)=P(q)g^\star(q)$. Hence, the nonholonomic construction preserves the approximation properties of the underlying regression method while enforcing the constraints exactly.

\section{Numerical example: Vertical rolling disk}
\label{sec:simulations}

To illustrate the proposed nonholonomic Gaussian process framework, we consider the vertical rolling disk, see \cite{bloch2003nonholonomic}. The configuration space is $Q = SE(2)\times S^1 \cong \mathbb{R}^2\times S^1\times S^1$ with coordinates $q=(x,y,\varphi,\theta)$, where $(x,y)$ denotes the contact point on the plane, $\varphi$ is the orientation of the disk in the plane, and $\theta$ is the rotation angle of the disk. 

The rolling-without-slipping constraints are
\begin{equation}
\dot x - R\cos\varphi\,\dot\theta = 0,
\qquad
\dot y - R\sin\varphi\,\dot\theta = 0.
\label{eq:vrd_constraints}
\end{equation}
These define the constraint distribution with constraint matrix
\begin{equation}
A(q)=
\begin{bmatrix}
1 & 0 & 0 & -R\cos\varphi\\
0 & 1 & 0 & -R\sin\varphi
\end{bmatrix}.
\label{eq:vrd_A}
\end{equation}

A convenient basis of the distribution is given by the vector fields $\displaystyle{X_1(q)=
\begin{bmatrix}
R\cos\varphi\,\, 
R\sin\varphi\,\,
0\,\,
1
\end{bmatrix}^{T}}$ and $X_2(q)=
\begin{bmatrix}
0\,\,
0\,\,
1\,\,
0
\end{bmatrix}^{T},$ so that every admissible velocity can be written as
$\dot q = B(q)\nu,
\,
B(q)=\begin{bmatrix} X_1(q) & X_2(q)\end{bmatrix},
\,
\nu\in\mathbb{R}^2$. By construction, $A(q)B(q)=0$.

We consider admissible dynamics of the form $\dot q = f(q)$, with $f(q)\in \D_q$. Rather than learning arbitrary vector fields in the ambient tangent bundle, we work with dynamics that are admissible through the representation $f(q)=B(q)\nu(q)$, where $\nu(q)\in\mathbb{R}^2$ is an adapted velocity field to $\D$.

We define a nominal model
$f_{\mathrm{nom}}(q)=B(q)\nu_{\mathrm{nom}}(q)$, and the perturbed dynamics
\[
f^\star(q)=B(q)\nu^\star(q),\quad
\nu^\star(q)=\nu_{\mathrm{nom}}(q)+\delta(q),
\]
where $\delta(q)$ is a smooth model disturbance. Thus, the true dynamics is obtained by augmenting the nominal field adapted to $\D$ with an unmodeled component, while admissibility is preserved by construction. Because both vector fields are generated through the same distribution basis $B(q)$, they satisfy $f_{\mathrm{nom}}(q),\, f^\star(q)\in \D_q
\, \text{for all } q\in Q$, so admissibility is preserved by construction.

In the simulation we use the parameters $R=1$, $\Omega=1.0$, $\omega=0.35$, $\varepsilon=0.18$, where $\Omega$ and $\omega$ denote the nominal (baseline) angular rates of the dynamics on adapted coordinates, and $\varepsilon$ sets the amplitude of the perturbation. To avoid an overly degenerate nominal trajectory, the nominal velocity field $\nu_{\mathrm{nom}}$ is chosen to be weakly state-dependent, rather than constant:
\begin{equation*}
\nu_{\mathrm{nom}}(q)=
\begin{bmatrix}
\Omega + 0.10\sin\varphi + 0.06\cos\theta\\[0.3em]
\omega + 0.08\cos\varphi - 0.05\sin\theta
\end{bmatrix}.
\label{eq:nu_nom}
\end{equation*}
The true velocity field is then defined by adding the smooth perturbation
\begin{equation*}
\delta(q)=
\varepsilon
\begin{bmatrix}
0.60\sin(\varphi-\theta)+0.25\cos(2\theta)\\[0.3em]
0.50\cos(\varphi+\theta)-0.20\sin(2\varphi)
\end{bmatrix},
\label{eq:delta_vrd}
\end{equation*}
so that
$\nu^\star(q)=\nu_{\mathrm{nom}}(q)+\delta(q)$, and therefore $f^\star(q)=B(q)\nu^\star(q)$ satisfies $f^\star(q)\in \D_q$ for all $q\in Q$.

Training data are generated from the true model by sampling states along three true trajectories with initial conditions
$(0,0,0.2,0.1)$, $(0,0,-0.6,0.4)$ and 
$(0,0,0.8,-0.5)$, integrated with time step $\Delta t=0.05$. The sampled states are then subsampled to obtain $N_{\mathrm{train}}=120$ training points.  Gaussian perturbations with standard deviation $\sigma_{\mathrm{state}}=0.05$ are added to the angular coordinates $(\varphi,\theta)$ of the sampled states, and noisy observations are generated as $y_i = f^\star(q_i) + \eta_i$,
$\eta_i \sim \mathcal{N}(0,\sigma^2 I)$,
with $\sigma=0.03$. The horizon time for testing is chosen as $T=25$.

We compare two Gaussian process models. The first is a standard vector-valued Gaussian process with kernel $K_{\mathrm{std}}(q,q') = k(q,q')I_4$,
which does not encode the nonholonomic constraints. The second is the proposed Gaussian process with the nonholonomic kernel
$K_{\mathrm{NH}}(q,q') = P(q)\,k(q,q')\,P(q')$, whose posterior mean predictor is denoted by $\hat f_{\mathrm{NH}}$, where $P(q)=I-A(q)^\dagger A(q)$ is the orthogonal projector onto $\D$. In this example, since $R=1$, the basis matrix is $\displaystyle{
B(q)=
\big[(\cos\varphi,\sin\varphi,0,1)^T,\ (0,0,1,0)^T\big]}$ and the projector can be written explicitly as $P(q)=B(q)\big(B(q)^\top B(q)\big)^{-1}B(q)^\top$, that is, 
\begin{equation}
\small
P(q)=
\begin{bmatrix}
\frac{\cos^2\varphi}{2} & \frac{\sin\varphi\cos\varphi}{2} & 0 & \frac{\cos\varphi}{2}\\
\frac{\sin\varphi\cos\varphi}{2} & \frac{\sin^2\varphi}{2} & 0 & \frac{\sin\varphi}{2}\\
0 & 0 & 1 & 0\\
\frac{\cos\varphi}{2} & \frac{\sin\varphi}{2} & 0 & \frac{1}{2}
\end{bmatrix}.\normalsize
\end{equation}

By the results established in Section~III, the nonholonomic kernel restricts the entire GP prior to admissible vector fields and therefore guarantees $A(q)\hat f_{\mathrm{NH}}(q)=0, \text{for all } q\in Q$.

In all experiments we use a squared-exponential scalar kernel acting on the angular variables $(\varphi,\theta)$,
\begin{equation}
k(q,q')=\sigma_f^2
\exp\!\left(
-\frac12
\begin{bmatrix}
\varphi-\varphi'\\
\theta-\theta'
\end{bmatrix}^{\!\top}
\Lambda^{-1}
\begin{bmatrix}
\varphi-\varphi'\\
\theta-\theta'
\end{bmatrix}
\right),
\label{eq:scalar_kernel_num}
\end{equation}
with $\Lambda=\mathrm{diag}(\ell_\varphi^2,\ell_\theta^2)
$.  Here $\sigma_f^2$ denotes the signal variance, $(\ell_\varphi,\ell_\theta)$ the length scales, and $\sigma_n^2$ the noise variance. The standard GP is trained directly in ambient coordinates $\mathbb{R}^4$, whereas the nonholonomic GP is trained in the adapted coordinates $\nu\in\mathbb{R}^2$ and mapped back to ambient space through $B(q)$.  The kernel hyperparameters obtained by marginal likelihood optimization are summarized in Table~\ref{tab:vrd_hyperparameters}.
\begin{table}[h]
\centering
\begin{tabularx}{\textwidth}{lcccc}
\hline
Model / channel & $\sigma_f^2$ & $\ell_\varphi$ & $\ell_\theta$ & $\sigma_n^2$ \\
\hline
Stand. GP, outp 1 & $0.949^2$ & 1.50 & 2.18 & 0.00198 \\
Stand. GP, output 2 & $1.49^2$  & 1.50 & 2.46 & 0.00409 \\
Stand. GP, output 3 & $1.21^2$  & 1.27 & 1.70 & 0.0939 \\
Stand. GP, output 4 & $0.989^2$ & 1.39 & 1.28 & 0.0583 \\
NH GP, adapted 1      & $1.02^2$  & 1.47 & 1.35 & 0.0306 \\
NH GP, adapted 2      & $1.21^2$  & 1.27 & 1.70 & 0.0939 \\
\hline
\end{tabularx}
\caption{Standard vs. nonholonomic GP\label{tab:vrd_hyperparameters}}
\end{table}
In this example the squared-exponential kernel is implemented in the angular chart $(\varphi,\theta)$. Since the trajectories considered here remain in a regime where this chart representation is numerically adequate, we use the corresponding Euclidean RBF kernel. A more intrinsic treatment on $S^1\times S^1$ could alternatively employ periodic kernels, see \cite{rasmussen2006gaussian}.

To assess model quality, we compare both local vector-field prediction and trajectory-level performance. First, we evaluate the pointwise prediction error
$e_f(q)=\|\hat f(q)-f^\star(q)\|$. Second, we measure violation of the nonholonomic constraints through $e_{\mathrm{nh}}(q)=\|A(q)\hat f(q)\|$. For the nonholonomic GP, this quantity vanishes identically up to numerical precision, whereas for the unconstrained GP it is generally nonzero. Third, we integrate the learned vector fields and compare the resulting trajectories with the true trajectory. If $q^\star(t)$ denotes the trajectory generated by $f^\star$ and $\hat q(t)$ the trajectory generated by a learned model, we use the planar tracking error as a physical observable:
\begin{equation}
\Delta(t)=\sqrt{(x^\star(t)-\hat x(t))^2+(y^\star(t)-\hat y(t))^2}.
\label{eq:trajectory_error_num}
\end{equation}

Fig.~\ref{fig:vrd_trajectory_comparison} displays the resulting trajectories. The nominal model exhibits a visible drift with respect to the true trajectory. The unconstrained GP already improves the prediction substantially, but the proposed nonholonomic GP yields the closest trajectory to the ground truth over the full horizon time. This is confirmed quantitatively in Fig.~\ref{fig:vrd_trajectory_error}, which reports the planar trajectory error $\Delta(t)$. The nonholonomic GP achieves the smallest error over essentially the entire simulation interval and improves upon both the nominal model and the unconstrained GP.

\begin{figure}[h!]
    \centering
    \includegraphics[width=.99\linewidth]{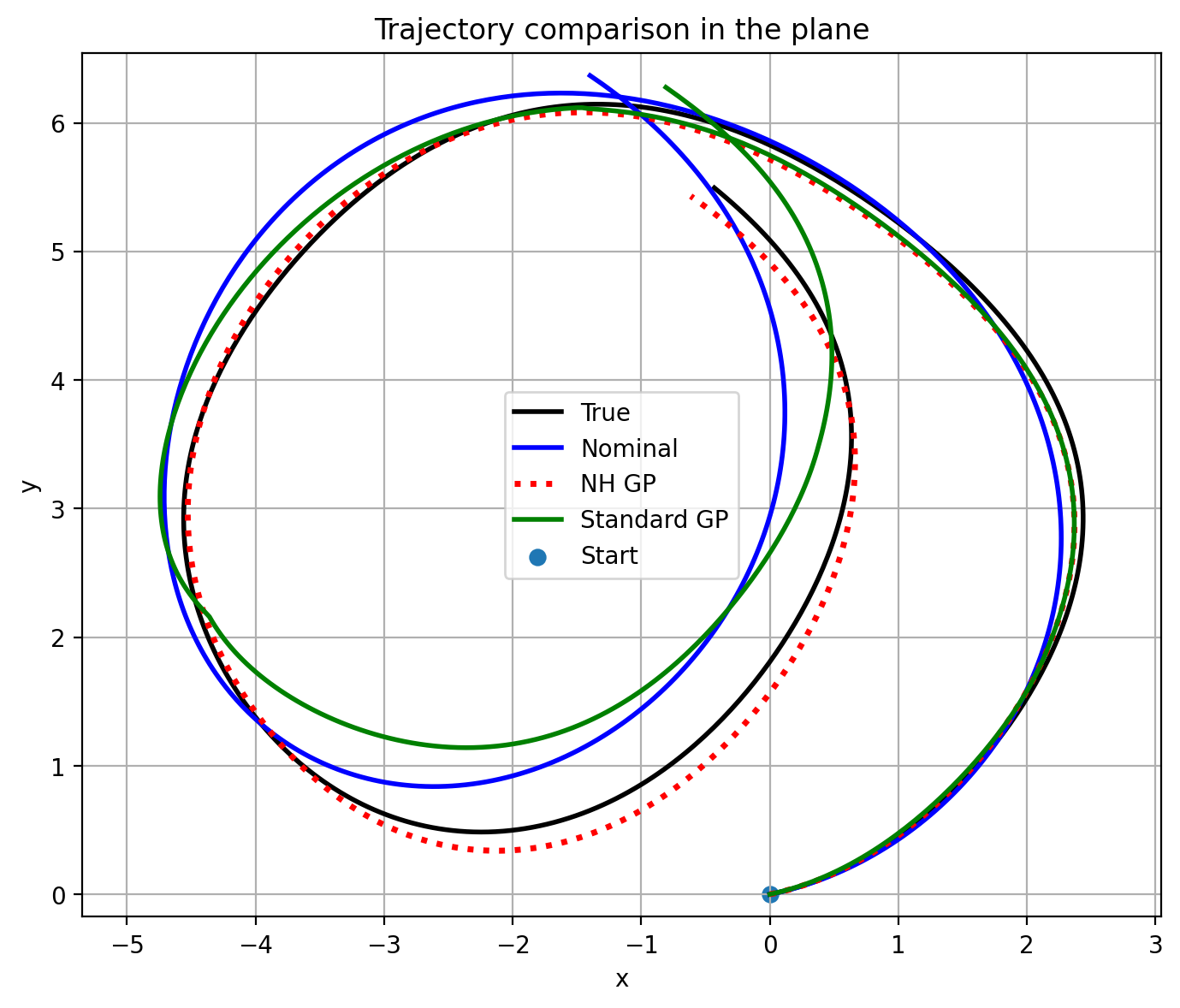}
    \caption{Trajectory comparison in the  plane: true dynamics, nominal model, nonholonomic GP, and standard GP.}
    \label{fig:vrd_trajectory_comparison}
\end{figure}

\begin{figure}[h!]
    \centering
    \includegraphics[width=.99\linewidth]{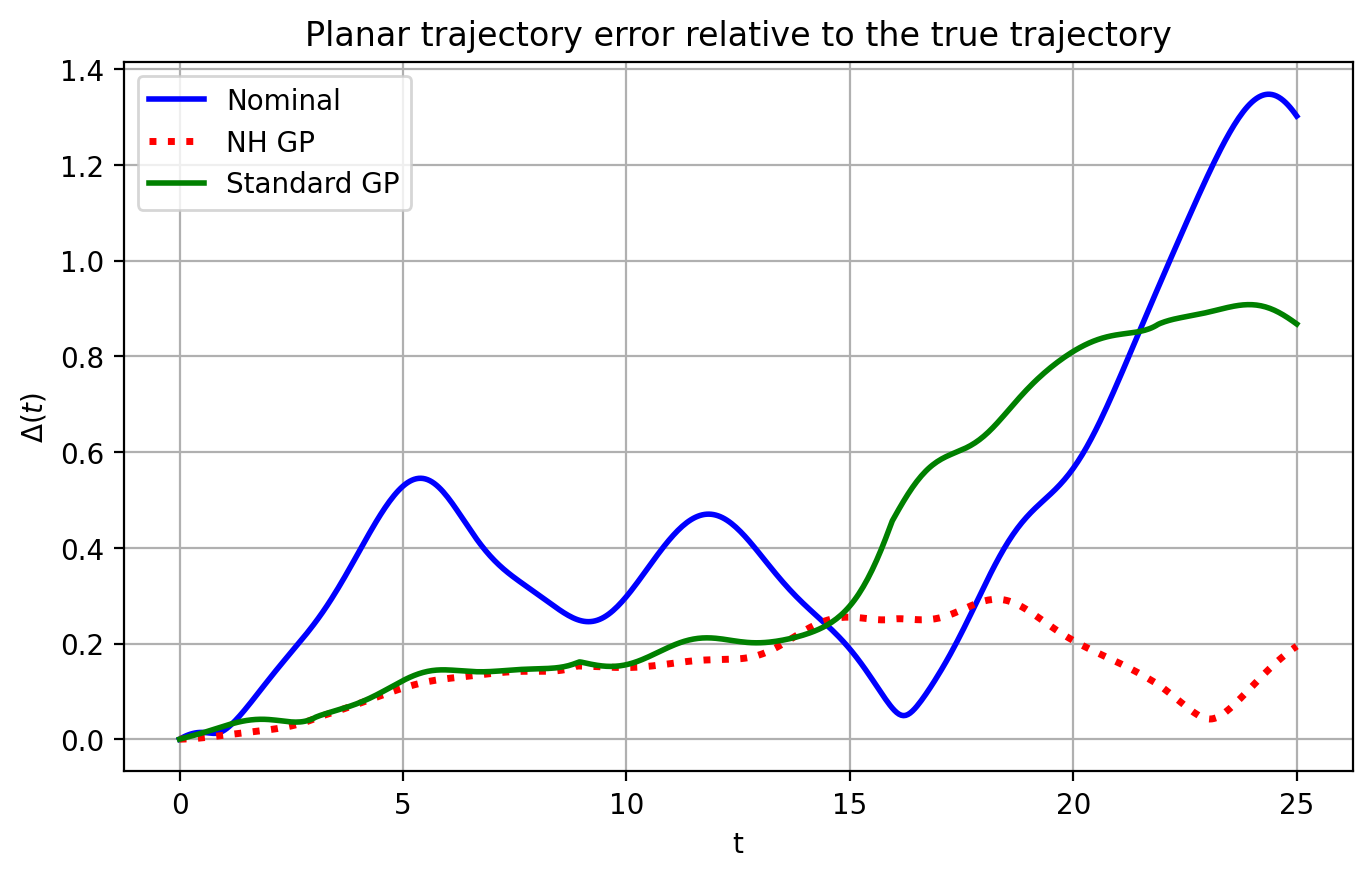}
    \caption{Planar trajectory error $\Delta(t)$ for the nominal, nonholonomic GP, and standard GP models.}
    \label{fig:vrd_trajectory_error}
\end{figure}

Fig.~\ref{fig:vrd_constraint_violation} shows the constraint violation metric $e_{\mathrm{nh}}(q)=\|A(q)\hat f(q)\|$ over the test set. As predicted by the theory, the nonholonomic GP preserves the rolling constraints exactly up to numerical precision, while the unconstrained GP produces nonzero components outside the admissible distribution. Finally, Fig.~\ref{fig:vrd_field_error} reports the pointwise field prediction error $e_f(q)=\|\hat f(q)-f^\star(q)\|$ on the test set. The nonholonomic GP also provides the smallest field error, showing that its geometric consistency does not come at the expense of predictive accuracy. On the contrary, in this example the nonholonomic prior improves both local field approximation and long-horizon trajectory prediction.

\begin{figure}[h!]
    \centering
    \includegraphics[width=.99\linewidth]{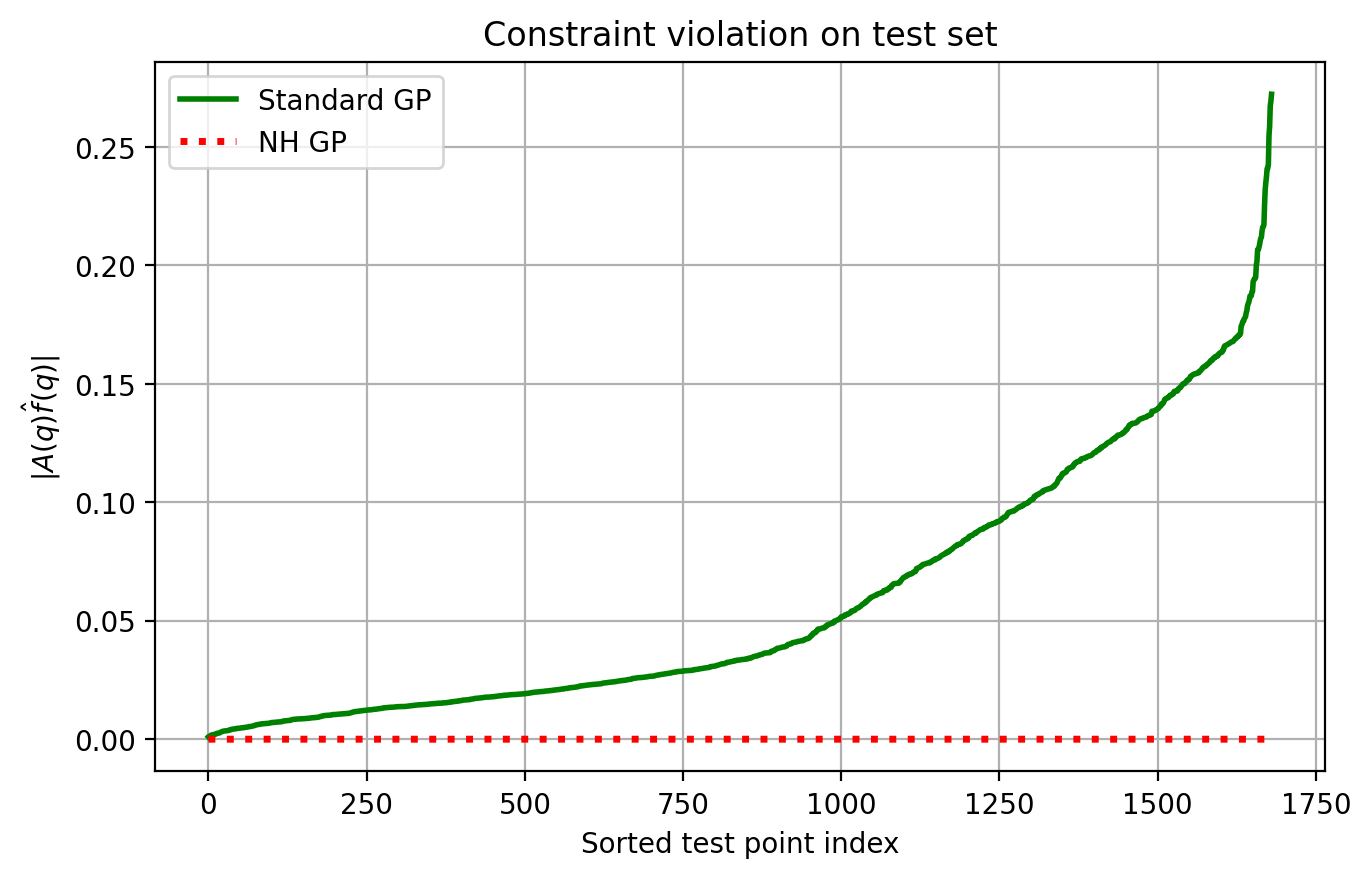}
    \caption{Constraint violation metric $e_{\mathrm{nh}}(q)=\|A(q)\hat f(q)\|$ for the standard and nonholonomic GP models.}
    \label{fig:vrd_constraint_violation}
\end{figure}

\begin{figure}[h!]
    \centering
    \includegraphics[width=.99\linewidth]{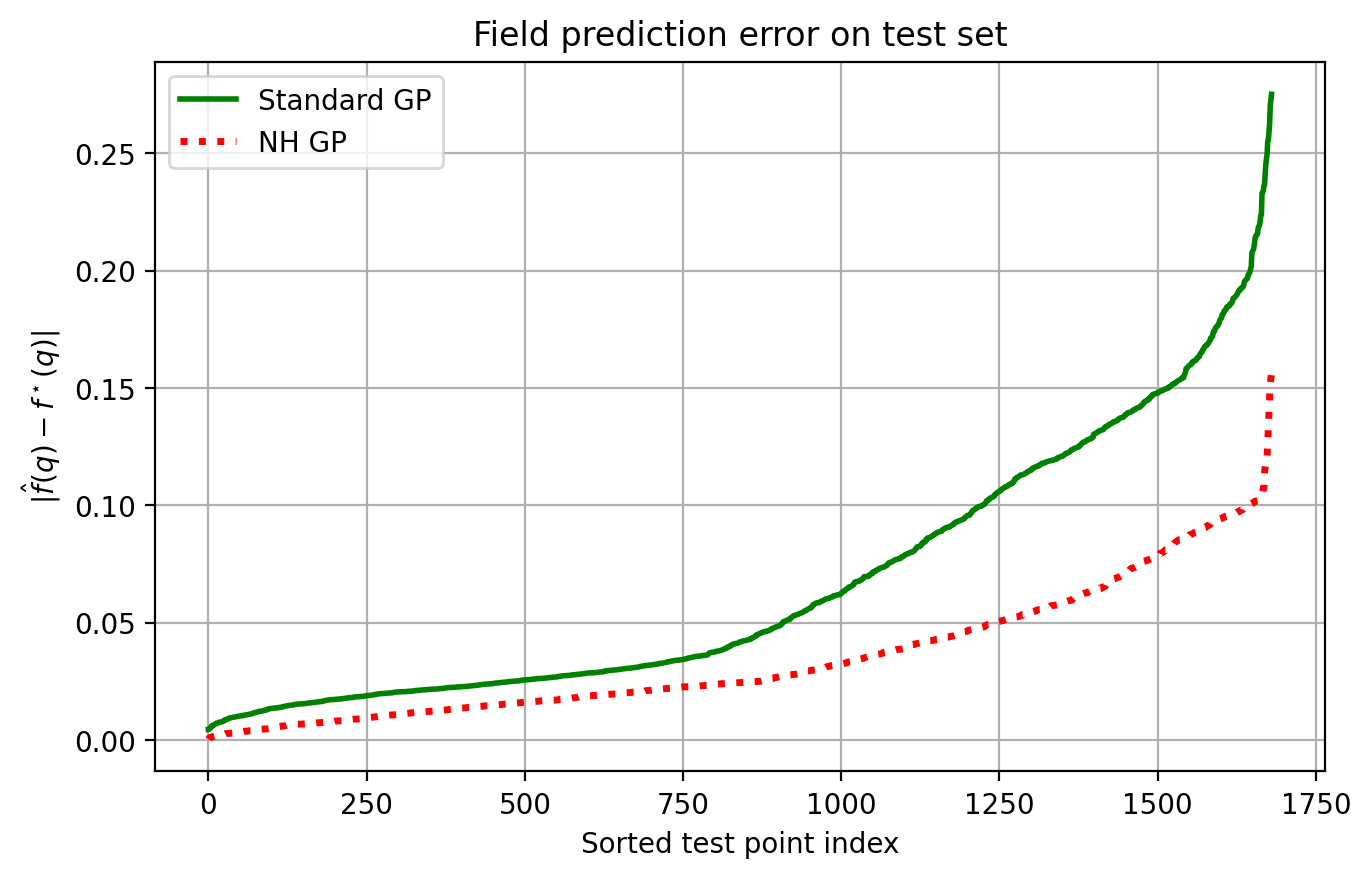}
    \caption{Pointwise field prediction error $e_f(q)=\|\hat f(q)-f^\star(q)\|$ on the test set.}
    \label{fig:vrd_field_error}
\end{figure}

Table~\ref{tab:vrd_metrics} shows the numerical results. The nonholonomic GP achieves the smallest mean field error, exactly preserves the constraint distribution, and yields the lowest mean and final trajectory errors among all models considered.

\begin{table}[h!]
\centering
\caption{Metrics used for the vertical rolling disk.}
\label{tab:vrd_metrics}
\setlength{\tabcolsep}{4pt}
\renewcommand{\arraystretch}{0.95}
\small
\begin{tabular}{lccc}
\hline
Metric & Nominal & Standard GP & NH GP \\
\hline
Mean field err.   & ---      & 0.067409                 & 0.035205 \\
Mean constr. viol.& ---      & $5.873109\times 10^{-2}$ & 0 \\
Max constr. viol. & ---      & $2.722539\times 10^{-1}$ & 0 \\
Mean planar err.  & 0.449343 & 0.373659                 & 0.148855 \\
Final planar err. & 1.302615 & 0.868002                 & 0.194430 \\
\hline
\end{tabular}
\end{table}

\section{Conclusions and Future Work}\label{sec:conclusions}

We introduced a nonholonomic kernel for GP regression that enforces the nonholonomic constraints exactly by projecting a scalar kernel onto the constraint distribution. The resulting construction yields valid vector-valued kernels, restricts the associated RKHS to admissible vector fields, and preserves consistency of the underlying regression scheme.

Future work will address learning on reduced nonholonomic systems, where the reconstruction process can make the estimation problem substantially more challenging. We also plan to incorporate volume variation as an additional metric in examples such as the Chaplygin sleigh and the Suslov problem, where volume is not necessarily preserved and may play an important role in the learning behavior.
\bibliographystyle{IEEEtran}
\bibliography{root}

@inproceedings{Greydanus2019,
  title={Hamiltonian Neural Networks},
  author={Greydanus, Sam and Dzamba, Misko and Yosinski, Jason},
  booktitle={Advances in Neural Information Processing Systems},
  volume={32},
  year={2019}
}

@article{Jin2020,
  title={SympNets: Intrinsic structure-preserving symplectic networks for identifying Hamiltonian systems},
  author={Jin, Pengzhan and Zhang, Zhen and Zhu, Aiqing and Tang, Yanzhi and Karniadakis, George Em},
  journal={Neural Networks},
  volume={132},
  pages={166--179},
  year={2020},
  publisher={Elsevier}
}

@inproceedings{Zhong2020,
  title={Symplectic ODE-Net: Learning Hamiltonian Dynamics with Control},
  author={Zhong, Yaofeng Desmond and Dey, Biswadip and Chakraborty, Amit},
  booktitle={International Conference on Learning Representations},
  year={2020}
}

@article{sriperumbudur2011universality,
  title={Universality, characteristic kernels and RKHS embedding of measures},
  author={Sriperumbudur, Bharath K. and Fukumizu, Kenji and Gretton, Arthur and Sch{\"o}lkopf, Bernhard and Lanckriet, Gert R. G.},
  journal={Journal of Machine Learning Research},
  volume={12},
  pages={2389--2410},
  year={2011}
}

@inproceedings{Cranmer2020,
  title={Lagrangian Neural Networks},
  author={Cranmer, Miles and Greydanus, Sam and Hoyer, Stephan and Battaglia, Peter and Spergel, David and Ho, Shirley},
  booktitle={ICLR 2020 Workshop on Integration of Deep Neural Models and Differential Equations}
}

@inproceedings{Finzi2020,
  title={Generalizing Convolutional Neural Networks for Equivariance to Lie Groups on Arbitrary Continuous Data},
  author={Finzi, Marc and Stanton, Samuel and Izmailov, Pavel and Wilson, Andrew Gordon},
  booktitle={Proceedings of the 37th International Conference on Machine Learning},
  series={Proceedings of Machine Learning Research},
  volume={119},
  pages={3165--3176},
  year={2020},
  publisher={PMLR}
}

@article{Vaquero2024,
  title={Symmetry Preservation in Hamiltonian Systems: Simulation and Learning},
  author={Vaquero, Miguel and Cort{\'e}s, Jorge and de Diego, David Mart{\'\i}n},
  journal={Journal of Nonlinear Science},
  volume={34},
  number={6},
  pages={115},
  year={2024},
  publisher={Springer}
}

@inproceedings{Borovitskiy2020,
  author    = {Viacheslav Borovitskiy and Alexander Terenin and Peter Mostowsky and Marc Peter Deisenroth},
  title     = {Mat{\'e}rn Gaussian Processes on Riemannian Manifolds},
  booktitle = {Advances in Neural Information Processing Systems},
  volume    = {33},
  pages     = {12426--12437},
  year      = {2020}
}

@article{Jayasumana2015,
  author    = {Sadeep Jayasumana and Richard Hartley and Mathieu Salzmann and Hongdong Li and Mehrtash Harandi},
  title     = {Kernel Methods on Riemannian Manifolds with Gaussian {RBF} Kernels},
  journal   = {IEEE Transactions on Pattern Analysis and Machine Intelligence},
  volume    = {37},
  number    = {12},
  pages     = {2464--2477},
  year      = {2015},
  doi       = {10.1109/TPAMI.2015.2414422}
}

@article{diaz2025lagrangian,
  title={Lagrangian neural networks for nonholonomic mechanics},
  author={D{\'\i}az, Viviana Alejandra and Salomone, Leandro Mart{\'\i}n and Zuccalli, Marcela},
  journal={Chaos, Solitons \& Fractals},
  volume={199},
  pages={116867},
  year={2025},
  publisher={Elsevier}
}

@article{Eldred2024,
  title={Lie--Poisson Neural Networks (LPNets): Data-based computing of Hamiltonian systems with symmetries},
  author={Eldred, Christopher and Gay-Balmaz, Fran{\c{c}}ois and Huraka, Sofiia and Putkaradze, Vakhtang},
  journal={Neural Networks},
  volume={173},
  pages={106162},
  year={2024},
  publisher={Elsevier}
}

@article{Wang2024,
  title={Learning Nonholonomic Dynamics with Constraint Discovery},
  author={Wang, Baiyue and Bloch, Anthony},
  journal={arXiv preprint arXiv:2410.15201},
  year={2024}
}

@book{ne_mark2004dynamics,
  title={Dynamics of nonholonomic systems},
  author={Neimark, Juru Isaakovich and Fufaev, Nikola Alekseevich},
  volume={33},
  year={2004},
  publisher={American Mathematical Soc.}
}

@article{carmeli2010vector,
  title={Vector valued reproducing kernel Hilbert spaces and universality},
  author={Carmeli, Claudio and De Vito, Ernesto and Toigo, Alessandro and Umanit{\'a}, Veronica},
  journal={Analysis and Applications},
  volume={8},
  number={01},
  pages={19--61},
  year={2010},
  publisher={World Scientific}
}

@article{micchelli2005learning,
  title={On learning vector-valued functions},
  author={Micchelli, Charles A and Pontil, Massimiliano},
  journal={Neural computation},
  volume={17},
  number={1},
  pages={177--204},
  year={2005},
  publisher={MIT Press}
}

@incollection{bloch2003nonholonomic,
  title={Nonholonomic mechanics},
  author={Bloch, Anthony M},
  booktitle={Nonholonomic mechanics and control},
  pages={},
  year={2003},
  publisher={Springer}
}

@book{monforte2002geometric,
  title={Geometric, control and numerical aspects of nonholonomic systems},
  author={ Cort{\'e}s, Jorge},
  number={1793},
  year={2002},
  publisher={Springer Science \& Business Media}
}

@inproceedings{beckers2022gaussian,
  title={Gaussian process port-Hamiltonian systems: Bayesian learning with physics prior},
  author={Beckers, Thomas and Seidman, Jacob and Perdikaris, Paris and Pappas, George J},
  booktitle={2022 IEEE 61st Conference on Decision and Control (CDC)},
  pages={1447--1453},
  year={2022},
  organization={IEEE}
}

@book{rasmussen2006gaussian,
  title={{Gaussian} processes for machine learning},
  author={Rasmussen, Carl Edward and Williams, Christopher KI},
  volume={1},
  year={2006},
  publisher={MIT press Cambridge}
}

@Book{steinwart2008support,
  title     = {Support vector machines},
  publisher = {Springer Science \& Business Media},
  year      = {2008},
  author    = {Steinwart, Ingo and Christmann, Andreas},
  doi={10.1007/978-0-387-77242-4},
}

\end{document}